\documentclass[12pt]{article}

\usepackage{amsthm}
\theoremstyle{plain}
\newtheorem{theorem}{Theorem}
\theoremstyle{definition}
\newtheorem{definition}{Definition}
\theoremstyle{proposition}

\theoremstyle{lemma}
\newtheorem{lemma}{Lemma}
\theoremstyle{remark}

\usepackage{amsmath}
\usepackage{amssymb}
\usepackage[dvips]{graphicx}
\usepackage{cite}

\makeatletter
 
  \@addtoreset{equation}{section}
 \makeatother

\begin{document}
\setlength{\oddsidemargin}{0cm}
\setlength{\baselineskip}{7mm}

\begin{titlepage}

~~\\

\vspace*{0cm}
    \begin{Large}
       \begin{center}
         {Moduli Space in Homological Mirror Symmetry}
       \end{center}
    \end{Large}
\vspace{1cm}

\begin{center}
           Matsuo S{\sc ato}\footnote
           {
e-mail address : msato@hirosaki-u.ac.jp}\\
      \vspace{1cm}
       
         {\it Department of Natural Science, Faculty of Education, Hirosaki University\\ 
 Bunkyo-cho 1, Hirosaki, Aomori 036-8560, Japan}

\end{center}

\hspace{5cm}

\begin{abstract}
\noindent
We prove that the moduli space of the pseudo holomorphic curves in the A-model on a symplectic torus is homeomorphic to a moduli space of Feynman diagrams in the configuration space of the morphisms in the B-model on the corresponding elliptic curve. These moduli spaces determine the $A_{\infty}$ structure of the both models.

\end{abstract}

\vfill
\end{titlepage}
\vfil\eject

\setcounter{footnote}{0}

\section{Introduction}
\setcounter{equation}{0}

In the one-dimensional homological mirror symmetry (HMS) \cite{HMS} , the A-model on a symplectic torus corresponds to the B-model on an elliptic curve \cite{PolishchukZaslow}. The objects, representing D-branes, are the Lagrangian submanifolds in the A-model and the complexes of the coherent sheaves in the B-model. In one dimension, any real one-dimensional submanifold is the Lagrangian and the line bundles are coherent sheaves. The morphisms, representing the open strings between the D-branes, are described by Abelian groups whose basis are given by intersecting points of the real one-dimensional submanifolds  in the A-model, and the maps between the complexes of the line bundles in the B-model. From these objects and morphisms, we can canonically construct the Fukaya category in the A-model and the derived category of coherent sheaves in the B-model. It is proved in \cite{FukayaOh, KontsevichSoibelman, Abouzaid, phD, Progress, TropicalBook, thetafunction, BraneMirror} that the Fukaya category is equivalent as an $A_{\infty}$-category to the differential graded (DG) category\footnote{The DG-category is an $A_{\infty}$-category where $m_d$ $(d \ge 3)$ are trivial.} canonically extended from the derived category of coherent sheaves. Furthermore, the Fukaya category is also equivalent as an $A_{\infty}$-category to a non-trivial $A_{\infty}$-category extended from the DG-category by using homotopy operators \cite{CSstring, Merkulov,PolishchukMassay, Polishchuk3pt, PolishchukA-Structures, AspinwallKatz, Zaslow, Kobayashi, Kanazawa}. In this extension, $m_3$ is explicitly constructed in the B-model \cite{Polishchuk3pt}, whereas $m_d$ $(d \ge 4)$ have not been explicitly constructed yet. 

In this paper, we extend the DG-category in a different way, based on the topological string amplitudes. $m_d$ are newly defined and explicitly constructed in the B-model. The $A_{\infty}$-category that consists of these $m_d$ is shown to be equivalent to the Fukaya category\footnote{Thus, the $A_{\infty}$-category defined in this paper is equivalent as an $A_{\infty}$-category to the $A_{\infty}$-category defined by the homotopy operators because they are equivalent to the Fukaya category as $A_{\infty}$ categories.}. In this construction, we find  a moduli space necessary to define $m_d$ that satisfy the $A_{\infty}$ relations in the B-model. This moduli space is homeomorphic to the moduli space of the pseudo holomorphic curves in the A-model.

\vspace{1cm}

\section{Topological String Amplitudes}
\setcounter{equation}{0}
The one-dimensional complex manifold in the B-model is an elliptic curve $E_{\tau}$ which is spanned by $1$ and $\tau \in \bold{C}$. There are two ways to generalize the derived category of coherent sheaves on $E_{\tau}$ to a differential graded (DG) category. One way is based on $\breve{C}$ech cohomology \cite{FukayaOh, KontsevichSoibelman, Abouzaid, phD, Progress, TropicalBook, thetafunction, BraneMirror} and the other is based on Dolbeault cohomology \cite{CSstring, PolishchukZaslow, Merkulov,PolishchukMassay, Polishchuk3pt, PolishchukA-Structures, AspinwallKatz, Zaslow}. They are equivalent by $\breve{C}$ech-Dolbeault isomorphism. We adopt Dolbeault cohomology. Coherent sheaves on $E_{\tau}$ are classified and constructed based on line bundles on $E_{\tau}$ \cite{Atiyah}. We study only the line bundles. The extension to general coherent sheaves is straight forward and the way is written in \cite{PolishchukZaslow, Polishchuk3pt}. The objects representing (d+1) D-branes are complexes of  line bundles of degree $n_i$, $L(n_i, u_i)$ ($i=0, \cdots, d$), where $u_i$ represent connections (gauge fields) over the D-branes. The morphisms representing open strings between $L(n_{i-1}, u_{i-1})$ and $L(n_i, u_i)$ are $H^{0,p}(E_{\tau}, L(n_{i-1, i}, u_{i-1, i}))$, where $p$ represent the degrees of the grading ($p=0,1$, $n_{i-1, i}:=n_i-n_{i-1}$ $u_{i-1, i}:=u_i-u_{i-1}$).

 When $n>0$, $H^{0,1}(E_{\tau}, L(n, u))=0$ and the elements of $H^{0,0}(E_{\tau}, L(n, u))$ are the basis of global sections of $L(n, u)$:
\begin{equation}
\theta
\begin{bmatrix}
\frac{m}{n} \\
u
\end{bmatrix}
(n z, n \tau)
\in
H^{0,0}(E_{\tau}, L(n, u)) \,\,\,
(m=0, \cdots, n-1),
\end{equation}
where $\theta
\begin{bmatrix}
a\\
b
\end{bmatrix}
(z, \tau)$
are theta functions with characteristics.

When $n<0$, $H^{0,0}(E_{\tau}, L(n, u))=0$ and the elements of $H^{0,1}(E_{\tau}, L(n, u))$ are the basis of harmonic (0,1)-forms with values in the dual bundle $L(n, u)$:
\begin{equation}
\overline{\theta
\begin{bmatrix}
\frac{m}{n} \\
-u
\end{bmatrix}
(-n z, -n \tau)}
\exp(2 \pi t (n z_2^2 +2 u_2 z_2)) d\bar{z}
\in
H^{0,1}(E_{\tau}, L(n, u)),
\end{equation}
where $\overline{\theta}$ is the complex conjugate of $\theta$, $\tau = A + i t, z= z_1+ \tau z_2, u=u_1+ \tau u_2$, and $A, t, z_1, z_2, u_1, u_2 \in \bold{R}$.

We are concerned with the $u=0$ case for simplicity. One can easily introduce the connections. We simplify the expressions as
\begin{equation}
\theta_{n}[p] (z,\tau):=
\theta
\begin{bmatrix}
p\\
0
\end{bmatrix}
(n z, n \tau).
\end{equation}
In this case, the strings are represented by 
\begin{eqnarray}
&&\theta_n [\frac{m}{n}](z, \tau)
\in
H^{0,0}(E_{\tau}, L(n)) \,\,\, (n > 0),
\nonumber \\
&&\theta^*_n [\frac{m}{n}](z, \bar{z}, \tau, \bar{\tau})d \bar{z}
:=
\overline{\theta_{-n} [\frac{m}{n}](z, \tau)}
\exp(2 \pi t n z_2^2) d\bar{z}
\in
H^{0,1}(E_{\tau}, L(n)) \,\,\, (n < 0).
\end{eqnarray}

Topological string amplitudes on $E_{\tau}$ are defined as \cite{WittenMirror, AspinwallKatz, BraneMirror}
\begin{equation}
\int_{E_{\tau}}
\Omega \wedge
\theta^*_n [\frac{m}{n}]d \bar{z}
\theta_{n_{i_1-1,i_1}} [\frac{m_{i_1-1,i_1}}{n_{i_1-1,i_1}}] \cdots \theta_{n_{i_c-1,i_c}} [\frac{m_{i_c-1,i_c}}{n_{i_c-1,i_c}}]
\psi_{n_{i_{c+1}-1,i_{c+1}}} [\frac{m_{i_{c+1}-1,i_{c+1}}}{n_{i_{c+1}-1,i_{c+1}}}] \cdots \psi_{n_{i_d-1,i_d}} [\frac{m_{i_d-1,i_d}}{n_{i_d-1,i_d}}], \label{amplitude1}
\end{equation}
where $n_{i_1-1,i_1}, \cdots, n_{i_c-1,i_c} >0$ and $n_{i_{c+1}-1,i_{c+1}}, \cdots, n_{i_d-1,i_d}, n <0$. $\Omega=dz$ is the holomorphic (1,0)-form, 
$\theta^*_n [\frac{m}{n}]d \bar{z}
\in H^{0,1}(E_{\tau}, L(n))$, and
$\theta_{n_{i_k-1,i_k}} [\frac{m_{i_k-1,i_k}}{n_{i_k-1,i_k}}]
\in H^{0,0}(E_{\tau}, L(n_{i_k-1,i_k}))$ $(k=1, \cdots, c)$. Here we explain what $\psi_{n} [\frac{m}{n}]$ is. In order to define topological string amplitudes of more than one $n<0$ states, we need to deform the theory by those states because more than one (0,1)-form cannot enter the topological string amplitudes in one dimension. The derived category describes such a deformed theory. In the topological string theory, the deformation by $n<0$ string states $O^{(0)}\in \Omega^{0,1}$ is given as follows. We define $O^{(1)}$ by $ \tilde{d} O^{(0)}= \{ Q, O^{(1)} \}$, where $\tilde{d}$ is a world sheet differential and $Q$ is a BRST operator. Then, the deformation of the theory is to insert $\psi:=\int_{\partial\Sigma} O^{(1)} \in \Omega^{0,0}$, where $\partial\Sigma$ is a world sheet boundary. 

In our case, we define this deformation by an isomorphism $\psi$: $\Omega^{0,1} \to \Omega^{0,0}$ by
\begin{equation}
O^{(0)}= \theta^*_n [\frac{m}{n}](z, \bar{z}, \tau, \bar{\tau})d \bar{z} 
\mapsto
\psi=\theta_n [\frac{m}{n}](z, \tau).
\end{equation}
That is, $\theta_n [\frac{m}{n}](z, \tau) \in \Omega^{0,0}$ represent string states, when not only $n>0$ but also $n<0$. This isomorphism will be justified later by mirror symmetry of the $A_{\infty}$ structure and $m_d$. By using this isomorphism, (\ref{amplitude1}) is written as
\begin{equation}
\int_{E_{\tau}}
\Omega \wedge
\theta^*_n [\frac{m}{n}]d \bar{z}
\theta_{n_{0,1}} [\frac{m_{0,1}}{n_{0,1}}] 
\theta_{n_{1,2}} [\frac{m_{1,2}}{n_{1,2}}] 
\cdots \theta_{n_{d-1, d}} [\frac{m_{d-1,d}}{n_{d-1, d}}].\label{amplitude2}
\end{equation}
On the other hand, (\ref{amplitude1}) should also be written by using $m_d$ in $A_{\infty}$-category \cite{WittenMirror, AspinwallKatz} like 
\begin{equation}
\sim
\int_{E_{\tau}}
\Omega \wedge
\theta^*_n [\frac{m}{n}]d \bar{z}
\,\,\,m_{d}(\frac{m_{0,1}}{n_{0,1}}, \cdots, \frac{m_{d-1,d}}{n_{d-1,d}}).
\label{amplitude2}
\end{equation}
Therefore, $m_d$ should be defined like
\begin{eqnarray}
&&m_{d}(\frac{m_{0,1}}{n_{0,1}}, \cdots, \frac{m_{d-1,d}}{n_{d-1,d}})
\nonumber \\
&\sim&
\theta_{n_{0,1}} [\frac{m_{0,1}}{n_{0,1}}] \cdots 
\theta_{n_{d-1, d}} [\frac{m_{d-1,d}}{n_{d-1, d}}].
\label{mdlike}
\end{eqnarray}
We will define completely $m_d$ that possesses $A_{\infty}$ structure in the next section. 

Here we discuss consistency of the integration over $E_{\tau}$ with the periodicity. Whereas the theta functions are invariant under $z \mapsto z+ 1$, they are transformed under $z \mapsto z+ \tau$ as
\begin{eqnarray}
\theta_n [\frac{m}{n}](z, \tau) 
&\mapsto&
e^{-\pi i n (2z+\tau)} \theta_n [\frac{m}{n}](z, \tau) \nonumber \\
\theta^*_n [\frac{m}{n}](z, \bar{z}, \tau, \bar{\tau})d \bar{z}
&\mapsto&
e^{-\pi i n (2z+\tau)} \theta^*_n [\frac{m}{n}](z, \bar{z}, \tau, \bar{\tau})d \bar{z}.
\end{eqnarray}
Then, the integrand is transformed as
\begin{eqnarray}
&&\theta^*_n [\frac{m}{n}]d \bar{z}
\theta_{n_{0,1}} [\frac{m_{0,1}}{n_{0,1}}] \cdots \theta_{n_{d-1, d}} [\frac{m_{d-1,d}}{n_{d-1, d}}] \nonumber \\
&\mapsto&
e^{-\pi i (2z+\tau)(n+n_{0,1}+ \cdots +n_{d-1,d})}
\theta^*_n [\frac{m}{n}]d \bar{z}
\theta_{n_{0,1}} [\frac{m_{0,1}}{n_{0,1}}] \cdots  \theta_{n_{d-1, d}} [\frac{m_{d-1,d}}{n_{d-1, d}}]. \nonumber
\end{eqnarray}
Because this should be invariant for periodicity, $n$ needs to be $-(n_{0,1}+ \cdots +n_{d-1,d})$ and then $n_{0,1}+ \cdots +n_{d-1, d}>0$.

\vspace{1cm}

\section{$m_d$ and $A_{\infty}$ Structure}
\setcounter{equation}{0}
In order to define $m_d$, we multiply theta functions with characteristics. They are defined by series as
\begin{definition}[theta functions with characteristics]
\begin{equation}
\theta
\begin{bmatrix}
a \\
b
\end{bmatrix}
(z, \tau)
:=
\sum_{m \in \bold{Z}}
\exp(\pi i (m+a)^2 \tau +2 \pi i (m+a)(z+b)),
\nonumber
\end{equation}
where
$a \in \bold{R}/\bold{Z}$, $b \in \bold{C}$.
\end{definition}
A product formula is given by
\begin{theorem}
\begin{eqnarray}
&&\theta
\begin{bmatrix}
\frac{a_1}{n_1} \\
0
\end{bmatrix}
(z_1, n_1\tau)
\,\,\,
\theta
\begin{bmatrix}
\frac{a_2}{n_2} \\
0
\end{bmatrix}
(z_2, n_2\tau) \nonumber \\ 
&=&
\sum_{d \in \bold{Z}/(n_1+n_2)\bold{Z}}
\theta
\begin{bmatrix}
\frac{n_2 a_1 - n_1 a_2 +n_1 n_2 d}{n_1 n_2 (n_1+n_2)} \\ 
0
\end{bmatrix}
(n_2 z_1 - n_1 z_2, n_1 n_2 (n_1+n_2) \tau) \,\,\,
\theta
\begin{bmatrix}
\frac{a_1 + a_2 +n_1 d}{n_1+n_2} \\
0
\end{bmatrix}
(z_1 + z_2, (n_1+n_2) \tau), 
\nonumber 
\end{eqnarray}
where $n_1, n_2 \in \bold{Z}$, $a_1, a_2 \in \bold{R}$, and $z_1, z_2, \tau \in \bold{C}$. 
\end{theorem}
While this formula was proved as an addition formula when $n_1, n_2 \in \bold{N}$ in \cite{Mumford}, it can also be proved as series when $n_1, n_2 \in \bold{Z}$ as follows.
\begin{proof}
\begin{eqnarray}
&&\mbox{(l.h.s.)} \nonumber \\ 
&=&
\sum_{m_1, m_2 \in \bold{Z}}
\exp(\pi i n_1 \tau (m_1 + \frac{a_1}{n_1})^2 
+2 \pi i (m_1+ \frac{a_1}{n_1}) z_1
+\pi i n_2 \tau (m_2 + \frac{a_2}{n_2})^2 
+2 \pi i (m_2+ \frac{a_2}{n_2}) z_2) \nonumber \\
&=&
\sum_{l, m \in \bold{Z}, d \in \bold{Z}/(n_1+n_2)\bold{Z}}
\exp \pi i( (n_1+n_2)\tau l^2 +2 ((a_1+a_2) \tau +z_1+z_2)l
+2 n_1 \tau l d + n_1 \tau d^2 \nonumber \\
&& \qquad \qquad \qquad \qquad \qquad 
 + (2 a_1 \tau +2z_1)d
+n_1 \tau (n_1+n_2)^2 m^2 + 2n_1(n_1+n_2) \tau l m \nonumber \\
&& \qquad \qquad \qquad \qquad \qquad 
+2 (a_1 \tau + z_1)(n_1+n_2) m 
+ 2n_1(n_1+n_2) \tau m d + \nonumber \\
&& \qquad \qquad \qquad \qquad \qquad 
\frac{a_1^2}{n_1}\tau +\frac{a_2^2}{n_2} \tau
+2\frac{a_1}{n_1} z_1 +2 \frac{a_2}{n_2} z_2), \label{lhs1}
\end{eqnarray}
where we have set
$m_1=l+p$, $m_2=l$ ($l, p \in \bold{Z}$), and 
$p=(n_1+n_2)m+d$ ($m \in \bold{Z}$, $d \in \bold{Z}/(n_1+n_2)\bold{Z}$). On the other hand,
\begin{eqnarray}
&&\mbox{(r.h.s.)}\nonumber \\ 
&=&
\sum_{m, n \in \bold{Z}, d \in \bold{Z}/(n_1+n_2) \bold{Z}}
\exp(\pi i n_1 n_2 (n_1+n_2) \tau (m+\frac{n_2 a_1 -n_1 a_2+n_1 n_2 d}{n_1 n_2 (n_1+n_2)})^2 \nonumber \\ 
&& \qquad \qquad \qquad \qquad \qquad
+ 2\pi i(m+\frac{n_2 a_1 -n_1 a_2+n_1 n_2 d}{n_1 n_2 (n_1+n_2)})(n_2 z_1-n_1 z_2))\nonumber \\ 
&& \qquad \qquad \qquad \qquad \qquad
+\pi i(n_1+n_2)\tau(n+\frac{a_1+a_2+n_1 d}{n_1+n_2})^2
\nonumber \\ 
&& \qquad \qquad \qquad \qquad \qquad
+2\pi i(n+\frac{a_1+a_2+n_1 d}{n_1+n_2})(z_1+z_2)).
\label{rhs1}
\end{eqnarray}
If we set $n=l+n_1 m$ ($l \in \bold{Z}$), (\ref{rhs1}) coincides with (\ref{lhs1}).
\end{proof}
In a special case: $z_1=n_1 z + u_1$ and $z_2=n_2 z + u_2$ $(u_1, u_2 \in \bold{C})$, we obtain 
\begin{eqnarray}
&&\theta
\begin{bmatrix}
\frac{a_1}{n_1} \\
u_1
\end{bmatrix}
(n_1 z, n_1\tau)
\,\,\,
\theta
\begin{bmatrix}
\frac{a_2}{n_2} \\
u_2
\end{bmatrix}
(n_2 z, n_2\tau) \nonumber \\ 
&=&
\sum_{d \in \bold{Z}/(n_1+n_2)\bold{Z}}
\theta
\begin{bmatrix}
\frac{n_2 a_1 - n_1 a_2 +n_1 n_2 d}{n_1 n_2 (n_1+n_2)} \\ 
n_2 u_1 -n_1 u_2
\end{bmatrix}
(0, n_1 n_2 (n_1+n_2) \tau) \,\,\,
\theta
\begin{bmatrix}
\frac{a_1 + a_2 +n_1 d}{n_1+n_2} \\
u_1 + u_2
\end{bmatrix}
((n_1+n_2)z, (n_1+n_2) \tau). \nonumber
\end{eqnarray}
This product is expanded by the theta functions with complex coefficients. The coefficients are independent of $z$. We simplify this formula. 
Because $\frac{a_1 + a_2 +n_1 d}{n_1+n_2} \in \bold{R}/\bold{Z}$ on the last line, we can add $-\frac{(n_1+n_2)}{(n_1+n_2)}d$ and obtain
\begin{eqnarray}
&&\theta
\begin{bmatrix}
\frac{a_1}{n_1} \\
u_1
\end{bmatrix}
(n_1 z, n_1\tau)
\,\,\,
\theta
\begin{bmatrix}
\frac{a_2}{n_2} \\
u_2
\end{bmatrix}
(n_2 z, n_2\tau) \nonumber \\ 
&=&
\sum_{d \in \bold{Z}/(n_1+n_2)\bold{Z}}
\theta
\begin{bmatrix}
\frac{n_2 a_1 - n_1 (a_2 -n_2 d)}{n_1 n_2 (n_1+n_2)} \\ 
n_2 u_1 -n_1 u_2
\end{bmatrix}
(0, n_1 n_2 (n_1+n_2) \tau) \,\,\,
\theta
\begin{bmatrix}
\frac{a_1 + (a_2 -n_2 d)}{n_1+n_2} \\
u_1 + u_2
\end{bmatrix}
((n_1+n_2)z, (n_1+n_2) \tau) \nonumber \\
&=&\sum_{m \in \bold{Z}, d \in \bold{Z}/(n_1+n_2) \bold{Z}}
\exp(\pi i n_1 n_2 (n_1+n_2) \tau (m+\frac{n_2 a_1 -n_1 (a_2-n_2 d)}{n_1 n_2 (n_1+n_2)})^2 \nonumber \\ 
&& \qquad \qquad \qquad \qquad \qquad
+ 2\pi i(m+\frac{n_2 a_1 -n_1 (a_2 -n_2 d)}{n_1 n_2 (n_1+n_2)})(n_2 u_1-n_1 u_2))\nonumber \\ 
&& \qquad \qquad \qquad 
\theta
\begin{bmatrix}
\frac{a_1 + (a_2 -n_2 d)}{n_1+n_2} \\
u_1 + u_2
\end{bmatrix}
((n_1+n_2)z, (n_1+n_2) \tau). \nonumber
\end{eqnarray}
Because $\frac{a_1 + (a_2 -n_2 d)}{n_1+n_2} \in \bold{R}/\bold{Z}$ on the last line, we can add $-\frac{n_2 (n_1+n_2) m}{(n_1+n_2)}d$. By defining $\alpha:=(n_1+n_2)m+d$ $(\alpha \in \bold{Z})$, we obtain
\begin{lemma}
\begin{eqnarray}
&&\theta
\begin{bmatrix}
\frac{a_1}{n_1} \\
u_1
\end{bmatrix}
(n_1 z, n_1\tau)
\,\,\,
\theta
\begin{bmatrix}
\frac{a_2}{n_2} \\
u_2
\end{bmatrix}
(n_2 z, n_2\tau) \nonumber \\ 
&=&\sum_{\alpha \in \bold{Z}}
\exp(\pi i \tau \frac{(n_2 a_1 -n_1 (a_2-n_2 \alpha))^2}{n_1 n_2 (n_1+n_2)}
+ 2\pi i(\frac{n_2 a_1 -n_1 (a_2 -n_2 \alpha)}{n_1 n_2 (n_1+n_2)})(n_2 u_1-n_1 u_2))\nonumber \\ 
&& \qquad  
\theta
\begin{bmatrix}
\frac{a_1 + (a_2 -n_2 \alpha)}{n_1+n_2} \\
u_1 + u_2
\end{bmatrix}
((n_1+n_2)z, (n_1+n_2) \tau).
\nonumber
\end{eqnarray}
Especially, when $u=0$ we obtain 
\begin{eqnarray}
&&\theta_{n_{1}}[\frac{a_{1}}{n_{1}}] (z,\tau)
\,\,\,
\theta_{n_{2}}[\frac{a_{2}}{n_{2}}] (z,\tau) \nonumber \\ 
&=&\sum_{\alpha \in \bold{Z}}
\exp(\pi i \tau \frac{(n_{2} a_{1} -n_{1} (a_{2}-n_{2} \alpha))^2}{n_{1} n_{2} (n_{1}+n_{2})}) 
\theta_{n_{1}+n_{2}}[\frac{a_{1} + (a_{2} -n_{2} \alpha)}{n_{1}+n_{2}}] (z,\tau).
\nonumber
\end{eqnarray}
\label{lemma1}
\end{lemma}
From now on, we abbreviate $\theta_{n}[p] (z,\tau)$ to $\theta_{n}[p]$ because $(z,\tau)$  does not vary. By using this formula, we obtain
\begin{theorem}
\begin{eqnarray}
&&\theta_{n_{0,1}}[\frac{a_{0,1}}{n_{0,1}}]
\theta_{n_{1,2}}[\frac{a_{1,2}}{n_{1,2}}] \cdots
\theta_{n_{d-1,d}}[\frac{a_{d-1, d}}{n_{d-1,d}}] \nonumber \\
&=&\sum_{\alpha_{1,2}, \cdots, \alpha_{d-1, d} \in \bold{Z}}
\exp(2 \pi i \tau S_d(\frac{\tilde{a}_{0,1}}{n_{0,1}}, \cdots, \frac{\tilde{a}_{d-1,d}}{n_{d-1,d}}))\theta_{n_{0,1}+ \cdots +n_{d-1,d}}[\frac{\tilde{a}_{0,1}+ \cdots + \tilde{a}_{d-1,d}}{n_{0,1}+ \cdots +n_{d-1,d}}], 
\nonumber
\end{eqnarray}
where 
\begin{eqnarray}
&&\tilde{a}_{i-1,i}:=a_{i-1,i}-n_{i-1,i}\alpha_{i-1,i} \,\,\,(i=1, \cdots, d) \nonumber \\
&&\alpha_{0,1} \equiv 0 \nonumber \\
&&S_d(p_{0,1}, \cdots, p_{d-1,d})
:=\frac{1}{2}((\sum_{i=1}^d n_{i-1,i} p_{i-1,i}(p_{i-1,i}-1)) 
-n_{0,d}p_{0,d}(p_{0,d}-1)) \nonumber \\
&&n_{0,d}:= n_{0,1}+ \cdots +n_{d-1,d} \nonumber \\
&&p_{0,d}:= \frac{1}{n_{0,d}}\sum_{i=1}^d n_{i-1,i} p_{i-1,i}. 
\nonumber
\end{eqnarray}
\label{producttheorem}
\end{theorem}
\begin{proof}
We start with $d=2$ case. Explicitly,
\begin{equation}
S_2(\frac{\tilde{a}_{0,1}}{n_{0,1}}, \frac{\tilde{a}_{1,2}}{n_{1,2}})
=
\frac{(n_{1,2} \tilde{a}_{0,1} -n_{0,1} \tilde{a}_{1,2})^2}{2 n_{0,1} n_{1,2} (n_{0,1}+n_{1,2})}.
\end{equation}
Then, Lemma \ref{lemma1} implies Theorem \ref{producttheorem} is right when $d=2$. For $2 \le k \le d-1$, explicitly, 
\begin{eqnarray}
&&S_{k+1}(\frac{\tilde{a}_{0,1}}{n_{0,1}}, \cdots,  \frac{\tilde{a}_{k,k+1}}{n_{k,k+1}})
-S_{k}(\frac{\tilde{a}_{0,1}}{n_{0,1}}, \cdots,  \frac{\tilde{a}_{k-1,k}}{n_{k-1,k}}) \nonumber \\
&=&
\frac{(n_{k,k+1} (\tilde{a}_{0,1}+ \cdots + \tilde{a}_{k-1,k}) -(n_{0,1}+ \cdots +n_{k-1,k}) \tilde{a}_{k,k+1})^2}{2 (n_{0,1}+ \cdots +n_{k-1,k}) n_{k,k+1} (n_{0,1}+ \cdots +n_{k-1,k}+n_{k,k+1})}.
\end{eqnarray}
Then, Lemma \ref{lemma1} implies 
\begin{eqnarray}
&&\theta_{n_{0,1}+ \cdots +n_{k-1,k}}[\frac{\tilde{a}_{0,1}+ \cdots + \tilde{a}_{k-1,k}}{n_{0,1}+ \cdots +n_{k-1,k}}]
\theta_{n_{k,k+1}}[\frac{a_{k,k+1}}{n_{k,k+1}}]  \nonumber \\
&=&\sum_{\alpha_{k+1} \in \bold{Z}}
\exp(2 \pi i \tau (S_{k+1}(\frac{\tilde{a}_{0,1}}{n_{0,1}}, \cdots, \frac{\tilde{a}_{k,k+1}}{n_{k,k+1}})-S_k(\frac{\tilde{a}_{0,1}}{n_{0,1}}, \cdots, \frac{\tilde{a}_{k-1,k}}{n_{k-1,k}})))\theta_{n_{0,1}+ \cdots +n_{k,k+1}}[\frac{\tilde{a}_{0,1}+ \cdots + \tilde{a}_{k,k+1}}{n_{0,1}+ \cdots +n_{k,k+1}}]. \nonumber
\end{eqnarray}
Therefore,
\begin{eqnarray}
&&\theta_{n_{0,1}}[\frac{a_{0,1}}{n_{0,1}}]
\theta_{n_{1,2}}[\frac{a_{1,2}}{n_{1,2}}] \cdots
\theta_{n_{d-1,d}}[\frac{a_{d-1, d}}{n_{d-1,d}}] \nonumber \\
&=&
\sum_{\alpha_{1,2}\in \bold{Z}}
\exp(2 \pi i \tau S_2(\frac{\tilde{a}_{0,1}}{n_{0,1}}, \frac{\tilde{a}_{1,2}}{n_{1,2}}))\theta_{n_{0,1}+n_{1,2}}[\frac{\tilde{a}_{0,1}+ \tilde{a}_{1,2}}{n_{0,1}+ n_{1,2}}] 
\theta_{n_{2,3}}[\frac{a_{2,3}}{n_{2,3}}] \cdots
\theta_{n_{d-1,d}}[\frac{a_{d-1, d}}{n_{d-1,d}}]
\nonumber \\
&=&
\sum_{\alpha_{1,2}\in \bold{Z}}
\exp(2 \pi i \tau S_2(\frac{\tilde{a}_{0,1}}{n_{0,1}}, \frac{\tilde{a}_{1,2}}{n_{1,2}}))
\sum_{\alpha_{2,3} \in \bold{Z}}
\exp(2 \pi i \tau (S_{3}(\frac{\tilde{a}_{0,1}}{n_{0,1}}, \frac{\tilde{a}_{1,2}}{n_{1,2}}, \frac{\tilde{a}_{2,3}}{n_{2,3}})-S_2(\frac{\tilde{a}_{0,1}}{n_{0,1}},  \frac{\tilde{a}_{1,2}}{n_{1,2}}))) \nonumber \\
&& \qquad \qquad\qquad\qquad\qquad\qquad\qquad
\theta_{n_{0,1}+n_{1,2}+n_{2,3}}[\frac{\tilde{a}_{0,1}+ \tilde{a}_{1,2} + \tilde{a}_{2,3}}{n_{0,1}+ n_{1,2} +n_{2,3}}] \theta_{n_{3,4}}[\frac{a_{3,4}}{n_{3,4}}] \cdots
\theta_{n_{d-1,d}}[\frac{a_{d-1, d}}{n_{d-1,d}}]
\nonumber \\
&=&
\sum_{\alpha_{1,2}, \alpha_{2,3} \in \bold{Z}}
\exp(2 \pi i \tau S_{3}(\frac{\tilde{a}_{0,1}}{n_{0,1}}, \frac{\tilde{a}_{1,2}}{n_{1,2}}, \frac{\tilde{a}_{2,3}}{n_{2,3}})) \theta_{n_{0,1}+n_{1,2}+n_{2,3}}[\frac{\tilde{a}_{0,1}+ \tilde{a}_{1,2} + \tilde{a}_{2,3}}{n_{0,1}+ n_{1,2} +n_{2,3}}] \theta_{n_{3,4}}[\frac{a_{3,4}}{n_{3,4}}] \cdots
\theta_{n_{d-1,d}}[\frac{a_{d-1, d}}{n_{d-1,d}}]
\nonumber \\
&=& \cdots
\nonumber \\
&=&
\sum_{\alpha_{1,2}, \cdots, \alpha_{d-1, d} \in \bold{Z}}
\exp(2 \pi i \tau S_d(\frac{\tilde{a}_{0,1}}{n_{0,1}}, \cdots, \frac{\tilde{a}_{d-1,d}}{n_{d-1,d}}))\theta_{n_{0,1}+ \cdots +n_{d-1,d}}[\frac{\tilde{a}_{0,1}+ \cdots + \tilde{a}_{d-1,d}}{n_{0,1}+ \cdots +n_{d-1,d}}]. 
\end{eqnarray}
\end{proof}
Next, we study $A_{\infty}$ structure. We define an extended theta function $\tilde{\theta}_n[\tilde{p}]$ whose configuration space is a universal cover $\bold{R}$ $(\ni \tilde{p})$ of the configuration space $\bold{R}/\bold{Z}$ $(\ni p)$ of the theta function $\theta_n[p]$. That is, $\tilde{\theta}_{n} [\frac{a}{n}+1] \neq \tilde{\theta}_{n} [\frac{a}{n}]$ whereas $\theta_{n} [\frac{a}{n}+1]=\theta_{n} [\frac{a}{n}]$. 
The product is defined as
\begin{definition}[product of $\tilde{\theta}$]
\begin{eqnarray}
&&\tilde{\theta}_{n_{0,1}}[\frac{a_{0,1}}{n_{0,1}}] (z,\tau)
\,\,\,
\tilde{\theta}_{n_{1,2}}[\frac{a_{1,2}}{n_{1,2}}] (z,\tau) \nonumber \\ 
&:=&
\exp(\pi i \tau \frac{(n_{1,2} a_{0,1} -n_{0,1} a_{1,2})^2}{n_{0,1} n_{1,2} (n_{0,1}+n_{1,2})}) 
\tilde{\theta}_{n_{0,1}+n_{1,2}}[\frac{a_{0,1} + a_{1,2}}{n_{0,1}+n_{1,2}}] (z,\tau) \label{universalproduct}. 
\nonumber
\end{eqnarray}
\end{definition}
This definition leads to 
\begin{lemma}
\begin{eqnarray}
&&\sum_{\alpha \in \bold{Z}}
\tilde{\theta}_{n_{0,1}}[\frac{a_{0,1}}{n_{0,1}}] (z,\tau)
\,\,\,
\tilde{\theta}_{n_{1,2}}[\frac{a_{1,2}}{n_{1,2}}-\alpha] (z,\tau) \nonumber \\ 
&=&
\sum_{\alpha \in \bold{Z}}
\exp(\pi i \tau \frac{(n_{1,2} a_{0,1} -n_{0,1} (a_{1,2}-n_{1,2} \alpha))^2}{n_{0,1} n_{1,2} (n_{0,1}+n_{1,2})}) 
\tilde{\theta}_{n_{0,1}+n_{1,2}}[\frac{a_{0,1} + (a_{1,2}-n_{1,2} \alpha)}{n_{0,1}+n_{1,2}}] (z,\tau). \nonumber
\end{eqnarray}
\label{coefficientlemma}
\end{lemma}
\begin{proof}
Clear.
\end{proof}
The coefficients in this formula coincide with those of Lemma \ref{lemma1}. Furthermore, the definition leads to 
\begin{lemma}[local product in the configuration space]
\begin{equation}
\tilde{\theta}_{n_{0,1}}[\tilde{p}] (z,\tau)
\,\,\,
\tilde{\theta}_{n_{1,2}}[\tilde{p}] (z,\tau)
=
\tilde{\theta}_{n_{0,1}+n_{1,2}}[\tilde{p}] (z,\tau).\nonumber
\end{equation} 
\label{locallemma}
\end{lemma}
\begin{proof}
\begin{equation}
\mbox{(l.h.s.)}
=
\exp(\pi i \tau \frac{(n_{1,2} n_{0,1} \tilde{p} -n_{0,1} n_{1,2} \tilde{p})^2}{n_{0,1} n_{1,2} (n_{0,1}+n_{1,2})}) 
\tilde{\theta}_{n_{0,1}+n_{1,2}}[\frac{n_{0,1} \tilde{p} + n_{1,2} \tilde{p}}{n_{0,1}+n_{1,2}}] (z,\tau) 
=
\mbox{(r.h.s.)}.
\end{equation}
\end{proof}
Next, we define a propagator $\tilde{G}_n[\tilde{p}_0, \tilde{p}](z,\tau)$ in the configuration space as
\begin{definition}[propagator in the configuration space]
\begin{equation}
\tilde{G}_n[\tilde{p}_0, \tilde{p}](z,\tau)\tilde{\theta}_{n}[\tilde{p}] (z,\tau)
:=
\tilde{\theta}_{n}[\tilde{p}_0] (z,\tau).
\nonumber
\end{equation}
\end{definition}
As a result, Feynman diagrams appear in the configuration space (Fig. \ref{diagram0}):
\begin{eqnarray}
\tilde{\theta}_{n_{0,1}}[\frac{a_{0,1}}{n_{0,1}}]
\tilde{\theta}_{n_{1,2}}[\frac{a_{1,2}}{n_{1,2}}]
&=&
\tilde{G}_{n_{0,1}}[\frac{a_{0,1}}{n_{0,1}}, \tilde{p}]
\tilde{G}_{n_{1,2}}[\frac{a_{1,2}}{n_{1,2}}, \tilde{p}]
\tilde{\theta}_{n_{0,1}}[\tilde{p}]
\tilde{\theta}_{n_{1,2}}[\tilde{p}] \nonumber \\
&=&
\tilde{G}_{n_{0,1}}[\frac{a_{0,1}}{n_{0,1}}, \tilde{p}]
\tilde{G}_{n_{1,2}}[\frac{a_{1,2}}{n_{1,2}}, \tilde{p}]
\tilde{\theta}_{n_{0,1}+n_{1,2}}[\tilde{p}].
 \nonumber \\
&=&
\tilde{G}_{n_{0,1}}[\frac{a_{0,1}}{n_{0,1}}, \tilde{p}]
\tilde{G}_{n_{1,2}}[\frac{a_{1,2}}{n_{1,2}}, \tilde{p}]
\tilde{G}_{n_{0,1}+n_{1,2}}[\tilde{p}, \tilde{p}']
\tilde{\theta}_{n_{0,1}+n_{1,2}}[\tilde{p}'].
\end{eqnarray}

\begin{figure}[htbp]
\begin{center}
\includegraphics[height=3cm, keepaspectratio, clip]{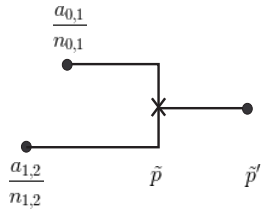}\end{center}
\caption{three point interaction}
\label{diagram0}
\end{figure}

When external states $\tilde{\theta}_{n}[\frac{m}{n}]$ propagate from $\frac{m}{n}$ to $\tilde{p}$, we can parametrize $\tilde{p}=\frac{m+v}{n}$. If two external states $\tilde{\theta}_{n_{0,1}}[\frac{m_{0,1}}{n_{0,1}}]$ and $\tilde{\theta}_{n_{1,2}}[\frac{m_{1,2}}{n_{1,2}}]$ propagate to the same point $\tilde{p}=\frac{m_{0,1}+v_{0,1}}{n_{0,1}}=\frac{m_{1,2}+v_{1,2}}{n_{1,2}}$ and interact locally, Lemma \ref{locallemma} leads to
\begin{equation}
\tilde{\theta}_{n_{0,1}}[\frac{m_{0,1}+v_{0,1}}{n_{0,1}}]
\tilde{\theta}_{n_{1,2}}[\frac{m_{1,2}+v_{1,2}}{n_{1,2}}] 
=
\tilde{\theta}_{n_{0,1}+n_{1,2}}[\frac{m_{0,1}+m_{1,2}+v_{0,1}+v_{1,2}}{n_{0,1}+n_{1,2}}],
\end{equation}
because $\tilde{p}=\frac{m_{0,1}+m_{1,2}+v_{0,1}+v_{1,2}}{n_{0,1}+n_{1,2}}$.  That is, $n$, $m$, and $v$ are preserved. 
Then, we define a canonical form of $\tilde{\theta}$, including internal states, as $\tilde{\theta}_{n}[\frac{m+v}{n}]$, where $n, m \in \bold{Z}$ and $v \in \bold{R}$. As a result, $m$ represent a kind of preserved numbers of string states. $\Delta \tilde{p}=\frac{\Delta v}{n}=\frac{v'-v}{n}$ represents how long the state propagates in an expression $\tilde{G}_{n}[\frac{m+v}{n}, \frac{m+v'}{n}]$.

We define the direction of the propagator $\tilde{G}_{n}[\frac{m+v}{n}, \frac{m+v'}{n}]$ as the same as the direction of $v$ and $v'$ ($v$ and $v'$ should have the same direction.). Because incoming states propagate from $\frac{m}{n}$ to $\frac{m+v}{n}$, $\Delta \tilde{p} =\frac{m+v}{n}-\frac{m}{n}=\frac{v}{n}.$ Coincidence of the signs of $\Delta \tilde{p}$ and $v$ implies that $n>0$ or $\Delta \tilde{p}=0$. That is, incoming states with $n<0$ cannot propagate. Similarly, outgoing states with $n>0$ cannot propagate. 

Strings between (i-1)-th and i-th D-branes can interact only with strings between (i-2)-th and (i-1)-th D-branes and strings between i-th and (i+1)-th D-branes. Therefore, we need to demand that the ordering of incoming states is non-commutative in the Feynman diagrams. 

We consider the moduli space $\mathcal{F}(\tilde{p}_{0,d}; \tilde{p}_{0,1}, \tilde{p}_{1,2}, \cdots, \tilde{p}_{d-1,d})$  of the Feynman diagrams that satisfy the above conditions,
for incoming states 
$\tilde{\theta}_{n_{0,1}}[\tilde{p}_{0,1}]$,
$\tilde{\theta}_{n_{1,2}}[\tilde{p}_{1,2}]$,
$\cdots$,
$\tilde{\theta}_{n_{d-1,d}}[\tilde{p}_{d-1,d}]$,
where
$\tilde{p}_{i-1,i}= \frac{\tilde{m}_{i-1,i}}{n_{i-1,i}}=\frac{m_{i-1,i}}{n_{i-1,i}}-\alpha_{i-1,i}$ 
$(\alpha_{i-1,i} \in \bold{Z}$, $i=1, \cdots, d)$,
and outgoing states 
$\tilde{\theta}_{n_{0,d}}[\tilde{p}_{0,d}]$,
where 
$\tilde{p}_{0,d}= \frac{\tilde{m}_{0,d}}{n_{0,d}}=\frac{\sum_{i=1}^d\tilde{m}_{i-1,i}}{n_{0,d}}$. We also consider the zero- and one-dimensional subspaces of the moduli space: $\mathcal{F}_0(\tilde{p}_{0,d}; \tilde{p}_{0,1}, \tilde{p}_{1,2}, \cdots, \tilde{p}_{d-1,d})$ and
$\mathcal{F}_1(\tilde{p}_{0,d}; \tilde{p}_{0,1}, \tilde{p}_{1,2}, \cdots, \tilde{p}_{d-1,d})$, respectively. Then, we obtain the following theorem. 
\begin{theorem}[correct Feynman diagram]
If $\mathcal{F}_0(\tilde{p}_{0,d}; \tilde{p}_{0,1}, \tilde{p}_{1,2}, \cdots, \tilde{p}_{d-1,d}) \neq \emptyset$, 

$\mathcal{F}_0(\tilde{p}_{0,d}; \tilde{p}_{0,1}, \tilde{p}_{1,2}, \cdots, \tilde{p}_{d-1,d})=\{\tilde{\delta}\}$, that is, $\mathcal{F}_0(\tilde{p}_{0,d}; \tilde{p}_{0,1}, \tilde{p}_{1,2}, \cdots, \tilde{p}_{d-1,d})$ consists of only one element $\tilde{\delta}$. Then, $\tilde{\delta} \in \mathcal{F}_0(\tilde{p}_{0,d}; \tilde{p}_{0,1}, \tilde{p}_{1,2}, \cdots, \tilde{p}_{d-1,d})$ determines a correlation function $<\tilde{p}_{0,d}; \tilde{p}_{0,1}, \tilde{p}_{1,2}, \cdots, \tilde{p}_{d-1,d} >_{\tilde{\delta}}$, which satisfies
\begin{equation}
\sum_{\alpha_{1,2}, \cdots, \alpha_{d-1, d}  \in \bold{Z}}<\tilde{p}_{0,d}; \tilde{p}_{0,1}, \tilde{p}_{1,2}, \cdots, \tilde{p}_{d-1,d} >_{\tilde{\delta}}
\theta_{n_{0, d}}[\tilde{p}_{0,d}]=
\theta_{n_{0,1}}[p_{0,1}]
\theta_{n_{1,2}}[p_{1,2}] \cdots
\theta_{n_{d-1,d}}[p_{d-1,d}]. 
\nonumber
\end{equation}
\label{correct}
\end{theorem}
\begin{proof}
First, we classify the Feynman diagrams. There is no internal vertex because dim($\mathcal{F}_0$)=0. Then, the vertices are external only: $\tilde{p}_{0,d}$, $\tilde{p}_{0,1}$, $\tilde{p}_{1,2}$, $\cdots$, $\tilde{p}_{d-1,d}$. $\tilde{p}_{0,1}$ is at the end of the diagrams because of the nearest neighbour interaction (Fig. \ref{diagram1}). $n_{0,1}>0$ because the string needs to propagate. Because of the three point interactions, $\tilde{p}_{i-1,i}$ not at the end of the diagrams cannot propagate, namely $n_{i-1,i}<0$ (Fig. \ref{diagram2}). Let us consider the outgoing vertex $\tilde{p}_{0,d}$. When $n_{0,d}<0$, a vertex propagate to $\tilde{p}_{0,d}$ (Fig. \ref{diagram3}). When $n_{0,d}>0$, because no vertex propagates to $\tilde{p}_{0,d}$, the last interaction must be as in (Fig. \ref{diagram4}).  As a result, we obtain only (i) (Fig. \ref{diagram5}) when $n_{0,d}<0$, and (ii) (Fig. \ref{diagram6}) when $n_{0,d}>0$. Therefore, we have shown that $\mathcal{F}_0(\tilde{p}_{0,d}; \tilde{p}_{0,1}, \tilde{p}_{1,2}, \cdots, \tilde{p}_{d-1,d})$ consists of only one element $\tilde{\delta}$ if $\mathcal{F}_0(\tilde{p}_{0,d}; \tilde{p}_{0,1}, \tilde{p}_{1,2}, \cdots, \tilde{p}_{d-1,d}) \neq \emptyset$.

\begin{figure}[htbp]
\begin{center}
\includegraphics[height=1cm, keepaspectratio, clip]{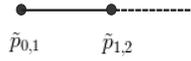}\end{center}
\caption{$\tilde{p}_{0,1}$}
\label{diagram1}
\end{figure}

\begin{figure}[htbp]
\begin{center}
\includegraphics[height=1cm, keepaspectratio, clip]{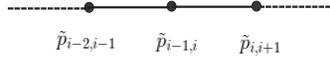}\end{center}
\caption{the external vertices that cannot propagate}
\label{diagram2}
\end{figure}

\begin{figure}[htbp]
\begin{center}
\includegraphics[height=1cm, keepaspectratio, clip]{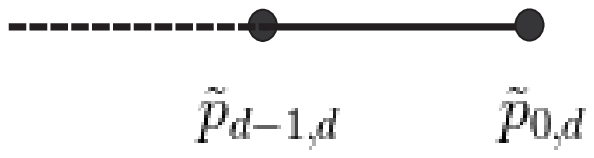}\end{center}
\caption{end point when $n_{0,d}<0$}
\label{diagram3}
\end{figure}

\begin{figure}[htbp]
\begin{center}
\includegraphics[height=3cm, keepaspectratio, clip]{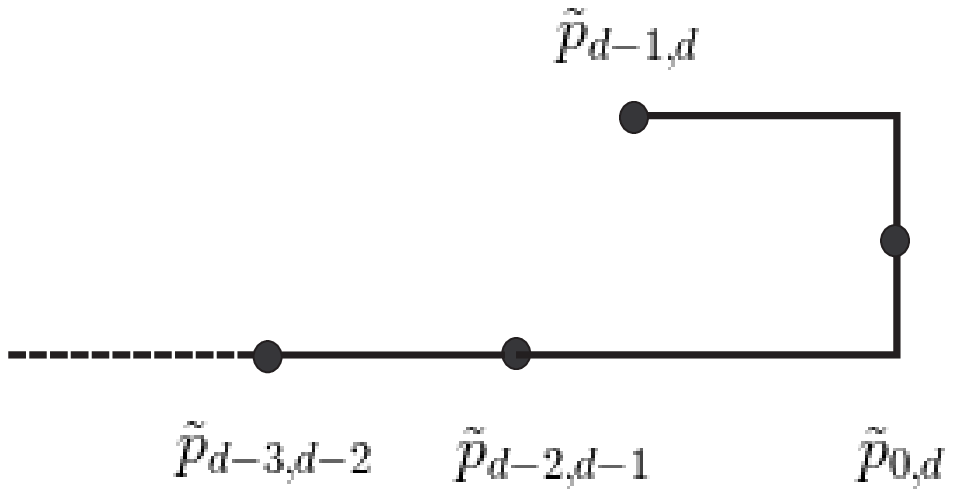}\end{center}
\caption{end point when $n_{0,d}>0$}
\label{diagram4}
\end{figure}

\begin{figure}[htbp]
\begin{center}
\includegraphics[height=1cm, keepaspectratio, clip]{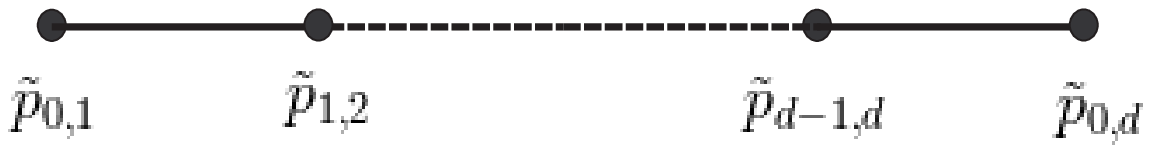}\end{center}
\caption{When $n_{0,d}<0$.}
\label{diagram5}
\end{figure}

\begin{figure}[htbp]
\begin{center}
\includegraphics[height=3cm, keepaspectratio, clip]{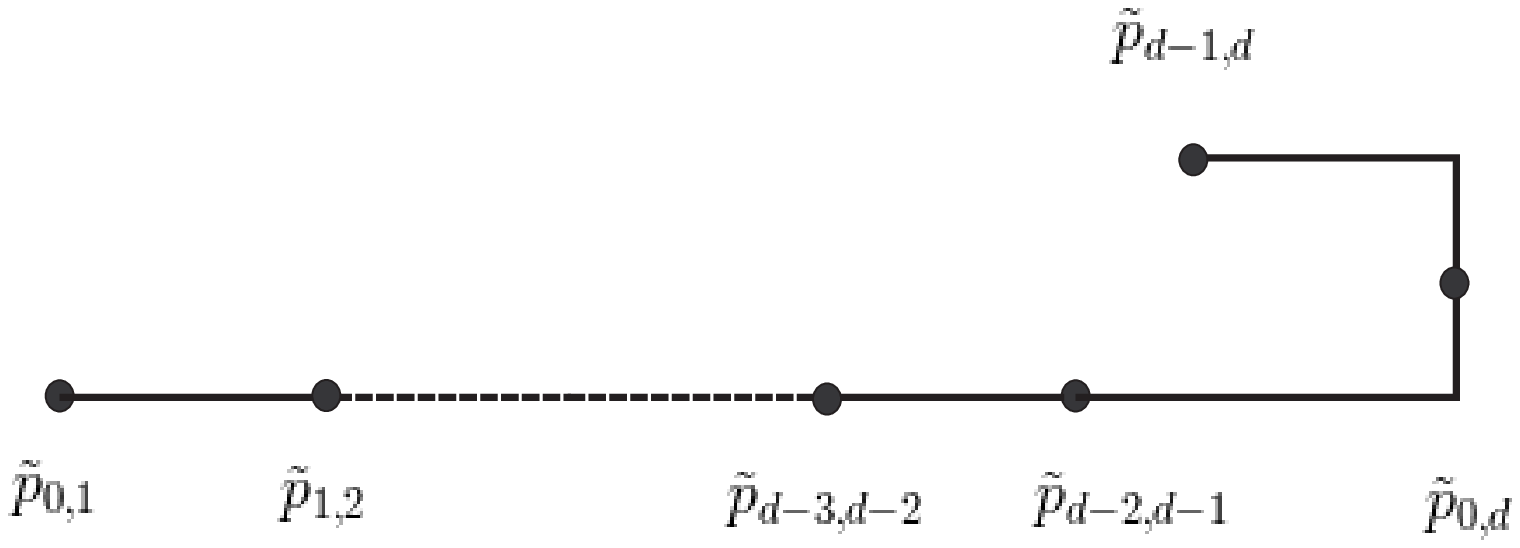}\end{center}
\caption{When $n_{0,d}>0$.}
\label{diagram6}
\end{figure}

Next, we calculate the correlation function determined by $\tilde{\delta}$.

(i) When $n_{0,d}<0$, from the diagram in Fig. \ref{diagram5},  
\begin{eqnarray}
&&\sum_{\alpha_{1,2}, \cdots, \alpha_{d-1, d}   \in \bold{Z}}<\tilde{p}_{0,d}; \tilde{p}_{0,1}, \tilde{p}_{1,2}, \cdots, \tilde{p}_{d-1,d} >_{\tilde{\delta}}
\tilde{\theta}_{n_{0,1}+ \cdots +n_{d-1, d}}[\tilde{p}_{0,d}]
\nonumber \\
&=&
\sum_{\alpha_{1,2}, \cdots, \alpha_{d-1, d}   \in \bold{Z}}
\tilde{G}_{n_{0,1}}[\frac{\tilde{m}_{0,1}}{n_{0,1}}, \frac{\tilde{m}_{1,2}}{n_{1,2}}]
\tilde{G}_{n_{0,1}+n_{1,2}}[\frac{\tilde{m}_{1,2}}{n_{1,2}}, \frac{\tilde{m}_{2,3}}{n_{2,3}}] 
\cdots
\tilde{G}_{n_{0,1}+n_{1,2}+\cdots+n_{d-2, d-1}}[\frac{\tilde{m}_{d-2, d-1}}{n_{d-2, d-1}}, \frac{\tilde{m}_{d-1,d}}{n_{d-1,d}}]
\nonumber \\
&& \qquad \qquad
\tilde{G}_{n_{0,1}+n_{1,2}+\cdots+n_{d-1,d}}[\frac{\tilde{m}_{d-1,d}}{n_{d-1,d}}, \frac{\tilde{m}_{0,d}}{n_{0,d}}]\tilde{\theta}_{n_{0,1}+n_{1,2}+\cdots+n_{d-1,d}}[\frac{\tilde{m}_{0,d}}{n_{0,d}}]
\nonumber \\
&=&
\sum_{\alpha_{1,2}, \cdots, \alpha_{d-1, d}   \in \bold{Z}}
\tilde{G}_{n_{0,1}}[\frac{\tilde{m}_{0,1}}{n_{0,1}}, \frac{\tilde{m}_{1,2}}{n_{1,2}}]
\tilde{G}_{n_{0,1}+n_{1,2}}[\frac{\tilde{m}_{1,2}}{n_{1,2}}, \frac{\tilde{m}_{2,3}}{n_{2,3}}] 
\cdots
\tilde{G}_{n_{0,1}+n_{1,2}+\cdots+n_{d-2, d-1}}[\frac{\tilde{m}_{d-2, d-1}}{n_{d-2, d-1}}, \frac{\tilde{m}_{d-1,d}}{n_{d-1,d}}]
\nonumber \\
&& \qquad \qquad
\tilde{\theta}_{n_{0,1}+n_{1,2}+\cdots+n_{d-1,d}}[\frac{\tilde{m}_{d-1,d}}{n_{d-1,d}}]
\nonumber \\
&=&
\sum_{\alpha_{1,2}, \cdots, \alpha_{d-1, d}   \in \bold{Z}}
\tilde{G}_{n_{0,1}}[\frac{\tilde{m}_{0,1}}{n_{0,1}}, \frac{\tilde{m}_{1,2}}{n_{1,2}}]
\tilde{G}_{n_{0,1}+n_{1,2}}[\frac{\tilde{m}_{1,2}}{n_{1,2}}, \frac{\tilde{m}_{2,3}}{n_{2,3}}] 
\cdots
\tilde{G}_{n_{0,1}+n_{1,2}+\cdots+n_{d-2, d-1}}[\frac{\tilde{m}_{d-2, d-1}}{n_{d-2, d-1}}, \frac{\tilde{m}_{d-1,d}}{n_{d-1,d}}]
\nonumber \\
&& \qquad \qquad
\tilde{\theta}_{n_{0,1}+n_{1,2}+\cdots+n_{d-2, d-1}}[\frac{\tilde{m}_{d-1,d}}{n_{d-1,d}}]
\tilde{\theta}_{n_{d-1,d}}[\frac{\tilde{m}_{d-1,d}}{n_{d-1,d}}]
\nonumber \\
&\cdots& 
\nonumber \\
&=&
\sum_{\alpha_{1,2}, \cdots, \alpha_{d-1, d}   \in \bold{Z}}
\tilde{\theta}_{n_{0,1}}[\frac{\tilde{m}_{0,1}}{n_{0,1}}]
\tilde{\theta}_{n_{1,2}}[\frac{\tilde{m}_{1,2}}{n_{1,2}}]
\cdots
\tilde{\theta}_{n_{d-1,d}}[\frac{\tilde{m}_{d-1,d}}{n_{d-1,d}}].
\end{eqnarray}
From Lemma \ref{coefficientlemma}, 
\begin{equation}
\sum_{\alpha_{1,2}, \cdots, \alpha_{d-1, d}  \in \bold{Z}}<\tilde{p}_{0,d}; \tilde{p}_{0,1}, \tilde{p}_{1,2}, \cdots, \tilde{p}_{d-1,d} >_{\tilde{\delta}}
\theta_{n_{0,1}+ \cdots +n_{d-1, d}}[\tilde{p}_{0,d}]=\theta_{n_{0,1}}[p_{0,1}]
\theta_{n_{1,2}}[p_{1,2}] \cdots
\theta_{n_{d-1,d}}[p_{d-1,d}]. 
\end{equation}

(ii) When $n_{0,d}>0$, from the diagram in Fig. \ref{diagram6},  
\begin{eqnarray}
&&\sum_{\alpha_{1,2}, \cdots, \alpha_{d-1, d}   \in \bold{Z}}<\tilde{p}_{0,d}; \tilde{p}_{0,1}, \tilde{p}_{1,2}, \cdots, \tilde{p}_{d-1,d} >_{\tilde{\delta}}
\tilde{\theta}_{n_{0,1}+ \cdots +n_{d-1, d}}[\tilde{p}_{0,d}]
\nonumber \\
&=&
\sum_{\alpha_{1,2}, \cdots, \alpha_{d-1, d}   \in \bold{Z}}
\tilde{G}_{n_{0,1}}[\frac{\tilde{m}_{0,1}}{n_{0,1}}, \frac{\tilde{m}_{1,2}}{n_{1,2}}]
\tilde{G}_{n_{0,1}+n_{1,2}}[\frac{\tilde{m}_{1,2}}{n_{1,2}}, \frac{\tilde{m}_{2,3}}{n_{2,3}}] 
\cdots
\nonumber \\
&& \qquad \qquad \qquad
\tilde{G}_{n_{0,1}+n_{1,2}+\cdots+n_{d-2,d-1}}[\frac{\tilde{m}_{d-2,d-1}}{n_{d-2,d-1}}, \frac{\tilde{m}_{0,d}}{n_{0,d}}]\tilde{\theta}_{n_{0,1}+n_{1,2}+\cdots+n_{d-2,d-1}}[\frac{\tilde{m}_{0,d}}{n_{0,d}}]
\nonumber \\
&& \qquad \qquad \qquad
\tilde{G}_{n_{d-1,d}}[\frac{\tilde{m}_{d-1,d}}{n_{d-1,d}}, \frac{\tilde{m}_{0,d}}{n_{0,d}}]\tilde{\theta}_{n_{d-1,d}}[\frac{\tilde{m}_{0,d}}{n_{0,d}}]
\nonumber \\
&=&
\sum_{\alpha_{1,2}, \cdots, \alpha_{d-1, d}   \in \bold{Z}}
\tilde{\theta}_{n_{0,1}}[\frac{\tilde{m}_{0,1}}{n_{0,1}}]
\tilde{\theta}_{n_{1,2}}[\frac{\tilde{m}_{1,2}}{n_{1,2}}]
\cdots
\tilde{\theta}_{n_{d-2,d-1}}[\frac{\tilde{m}_{d-2,d-1}}{n_{d-2,d-1}}]
\tilde{\theta}_{n_{d-1,d}}[\frac{\tilde{m}_{d-1,d}}{n_{d-1,d}}].
\end{eqnarray}
From Lemma \ref{coefficientlemma},
\begin{equation}
\sum_{\alpha_{1,2}, \cdots, \alpha_{d-1, d}  \in \bold{Z}}<\tilde{p}_{0,d}; \tilde{p}_{0,1}, \tilde{p}_{1,2}, \cdots, \tilde{p}_{d-1,d} >_{\tilde{\delta}}
\theta_{n_{0,1}+ \cdots +n_{d-1, d}}[\tilde{p}_{0,d}]=\theta_{n_{0,1}}[p_{0,1}]
\theta_{n_{1,2}}[p_{1,2}] \cdots
\theta_{n_{d-1,d}}[p_{d-1,d}]. 
\end{equation}
\end{proof}
By using these Feynman diagrams, we define $m_d$.
\begin{definition}[$m_d$]
\begin{eqnarray}
&&m_d(p_{0,1},p_{1,2}, \cdots, p_{d-1,d}) \nonumber \\
&:=&\sum_{\alpha_{1,2}, \cdots, \alpha_{d-1, d}  \in \bold{Z}}
\sum_{\tilde{\delta} \in \mathcal{F}_0(\tilde{p}_{0,d}; \tilde{p}_{0,1},\tilde{p}_{1,2}, \cdots, \tilde{p}_{d-1,d})} 
(-1)^{|\tilde{\delta}|}<\tilde{p}_{0,d}; \tilde{p}_{0,1}, \tilde{p}_{1,2}, \cdots, \tilde{p}_{d-1,d} >_{\tilde{\delta}}
\theta_{n_{0, d}}[\tilde{p}_{0,d}] \nonumber
\end{eqnarray}
\label{defmd}
\end{definition}
$|\tilde{\delta}| \in \bold{Z}$ should be determined so that $m_d$ satisfy the $A_{\infty}$ relation. We have given a general definition including more than one-dimensional case, whereas $\tilde{\delta}$ is unique in this one-dimensional situation. Although $m_d$ are essentially products of the theta functions as one can see in Theorem \ref{correct}, the naturally defined directions of the propagators and non-commutativity of the external vertices restrict the states of the external vertices. As a result, we obtain the following theorems.
\begin{theorem}
The degree of $m_d$ is 2-d.
\end{theorem}
\begin{proof}
Because $p$ in $H^{0,p}(E_{\tau}, L(n_{i-1, i}, u_{i-1, i}))$ represent the degrees of the open string states, the states with $n_{i-1, i}>0$ and $n_{i-1, i}<0$ have degrees 0 and 1, respectively. In the proof of Theorem \ref{correct}, the Feynman diagrams are classified into two kinds (Fig. \ref{diagram5}, \ref{diagram6}). In the diagrams in Fig. \ref{diagram5}, $n_{0,d}<0$, $n_{0,1}>0$, and $n_{i-1,i}<0$ ($i=2, \cdots d$). If we compare the degrees of the both sides of the formula in Definition \ref{defmd}, we obtain $0+(d-1)+\mbox{(the degree of $m_d$)}=1$. Thus, the degree of $m_d$ is $2-d$. Because $n_{0,d}>0$, $n_{0,1}>0$, $n_{i-1,i}<0$ ($i=2, \cdots d-1$), and $n_{d-1,d}>0$ in the diagrams in Fig. \ref{diagram6}, $0+(d-2)+0+\mbox{(the degree of $m_d$)}=0$. Thus, the degree of $m_d$ is $2-d$.
\end{proof}

\begin{theorem}[$A_{\infty}$ relation]
\begin{eqnarray}
&&\sum_{1 \le p \le d, 0 \le q \le d-p}
(-1)^{deg(p_{0,1})+\cdots+deg(p_{q-1,q})-q} \nonumber \\
&& m_{d-p+1} (p_{0,1}, \cdots, p_{q-1,q}, m_p(p_{q, q+1}, \cdots, p_{p+q-1, p+q}), p_{p+q, p+q+1}, \cdots, p_{d-1,d})=0,
\nonumber
\end{eqnarray}
where $deg(p_{i-1,i})$ $(i=1, \cdots, d)$ represent the degrees of the string states.  
\end{theorem}

\begin{proof}
\begin{eqnarray}
&&m_{d-p+1}(p_{0,1}, \cdots, p_{q-1,q}, m_p(p_{q, q+1}, \cdots, p_{p+q-1, p+q}), p_{p+q, p+q+1}, \cdots, p_{d-1,d}) 
\nonumber \\
&=&
\sum_{\alpha_{q+1, q+2},\cdots, \alpha_{p+q-1, p+q}  \in \bold{Z}}
\sum_{\tilde{\delta}_1 \in \mathcal{F}_0(\tilde{r}; \tilde{p}_{q, q+1},  \cdots, \tilde{p}_{p+q-1, p+q})} 
(-1)^{|\tilde{\delta_1}|}
<\tilde{r}; \tilde{p}_{q, q+1},  \cdots, \tilde{p}_{p+q-1, p+q} >_{\tilde{\delta}_1}
\nonumber \\
&& \qquad \qquad \qquad \qquad \qquad \qquad \qquad \qquad  
m_{d-p+1}(p_{0,1}, \cdots, p_{q-1,q}, \tilde{r},  
p_{p+q, p+q+1}, \cdots, p_{d-1,d}) \nonumber \\
&=&
\sum_{\alpha_{0,1}, \cdots, \alpha_{q-1, q}, \alpha_{q+1, q+2}, \cdots, \alpha_{d-1, d}  \in \bold{Z}}
\sum_{\tilde{\delta}_1 \in \mathcal{F}_0(\tilde{r}; \tilde{p}_{q, q+1},  \cdots, \tilde{p}_{p+q-1, p+q})} 
\sum_{\tilde{\delta}_2 \in \mathcal{F}_0(\tilde{p}_{0,d}; \tilde{p}_{0,1},  \cdots, \tilde{p}_{q-1,q}, \tilde{r},  
\tilde{p}_{p+q, p+q+1}, \cdots, \tilde{p}_{d-1,d})}
\nonumber \\
&& 
(-1)^{|\tilde{\delta_1}|+|\tilde{\delta_2}|}
<\tilde{r}; \tilde{p}_{q, q+1},  \cdots, \tilde{p}_{p+q-1, p+q} >_{\tilde{\delta}_1}
<\tilde{p}_{0,d}; \tilde{p}_{0,1},  \cdots, \tilde{p}_{q-1,q}, \tilde{r}, 
\tilde{p}_{p+q, p+q+1}, \cdots, \tilde{p}_{d-1,d} >_{\tilde{\delta}_2}
\nonumber \\
&& 
\theta_{n_{0,1}+ \cdots +n_{d-1, d}}[\tilde{p}_{0,d}]. 
\label{lhsofA}
\end{eqnarray}
When $\tilde{r}$ take arbitrary values, the diagrams in Fig. \ref{diagram7} are elements of $\bar{\mathcal{F}}_1(\tilde{p}_{0,d}; \tilde{p}_{0,1}, \tilde{p}_{1,2}, \cdots, \tilde{p}_{d-1,d})$, which is the one-dimensional subspace of $\bar{\mathcal{F}}(\tilde{p}_{0,d}; \tilde{p}_{0,1}, \tilde{p}_{1,2}, \cdots, \tilde{p}_{d-1,d})$ where $\bar{\mathcal{F}}(\tilde{p}_{0,d}; \tilde{p}_{0,1}, \tilde{p}_{1,2}, \cdots, \tilde{p}_{d-1,d})$ is a closure of $\mathcal{F}(\tilde{p}_{0,d}; \tilde{p}_{0,1}, \tilde{p}_{1,2}, \cdots, \tilde{p}_{d-1,d})$. 
Especially, $\tilde{\delta}_1 \otimes \tilde{\delta}_2$ are  the diagrams in Fig. \ref{diagram7} with $\tilde{r}=\frac{n_{q,q+1}\tilde{p}_{q,q+1}+\cdots+n_{q+p-1,q+p}\tilde{p}_{q+p-1,q+p}}{n_{q,q+1}+\cdots+n_{q+p-1,q+p}}$, which are elements of $\partial \bar{\mathcal{F}}_1(\tilde{p}_{0,d}; \tilde{p}_{0,1}, \tilde{p}_{1,2}, \cdots, \tilde{p}_{d-1,d})$, where
$\partial \bar{\mathcal{F}}_1 (\tilde{p}_{0,d}; \tilde{p}_{0,1}, \tilde{p}_{1,2}, \cdots, \tilde{p}_{d-1,d})$ is a boundary of $\bar{\mathcal{F}}_1(\tilde{p}_{0,d}; \tilde{p}_{0,1}, \tilde{p}_{1,2}, \cdots, \tilde{p}_{d-1,d})$.  By choosing $|\tilde{\delta}|$ appropriately, (\ref{lhsofA}) becomes 0 because the contributions from the boundary of the one-dimensional space cancel with each other as in the same mechanism in the A-model \cite{FukayaCategoryBook, phD, Progress, BraneMirror}.  
\end{proof}

\begin{figure}[htbp]
\begin{center}
\includegraphics[height=9cm, keepaspectratio, clip]{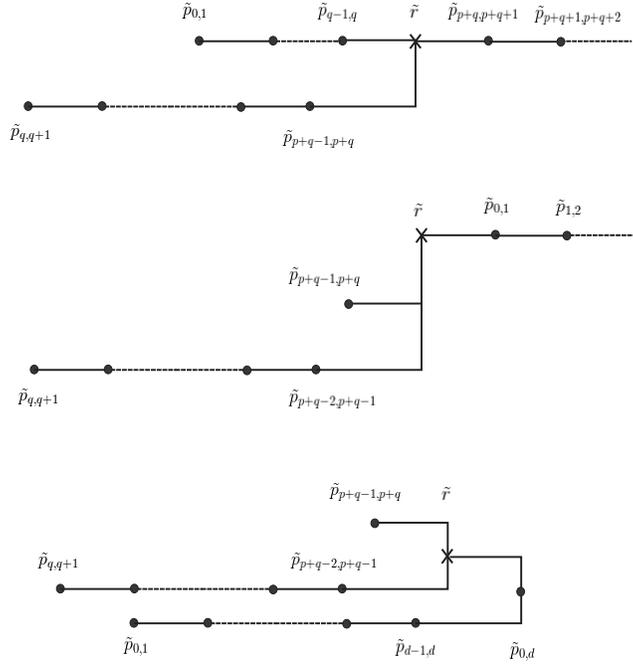}\end{center}
\caption{$\tilde{\delta}_1 \otimes \tilde{\delta}_2 \in \partial \bar{\mathcal{F}}_1$}
\label{diagram7}
\end{figure}

Here we compare $m_d$ in the B-model and A-model.
\begin{theorem}[mirror symmetry]
$m_d$ in the B-model in Definition \ref{defmd} coincides with $m_d$ in the A-model by identifying $\theta_n[\tilde{p}]$ with $[\tilde{p}]$\footnote{
This identification is a quasi-isomorphism \cite{phD, Progress, BraneMirror}.
}.
\end{theorem}
\begin{proof}
The formula in Definition \ref{defmd} can be regarded as $m_d$ in the A-model by replacing $\theta_n[\tilde{p}]$ with $[\tilde{p}]$. The following is the definition of the formula in Definition \ref{defmd} in the A-model \cite{phD, Progress, BraneMirror}. In the A-model, the i-th D-brane (Lagrangian submanifold) $\tilde{L}_i$ $(i=1, \cdots, d)$ is represented by a line with slope $-n_i$ on the universal cover of the symplectic torus $TB$. The intersection of $\tilde{L}_{i-1}$ and $\tilde{L}_i$ is represented by $\tilde{y}=\tilde{p}_{i-1, i}$ in $(\tilde{y}, \tilde{x})$ coordinates that is the universal cover of $(y,x)$ coordinates of $TB$ $(0 \le y,x <1)$. $[\tilde{p}_{i-1, i}]$ represents the string state with $n_{i-1, i}$ between the D-branes $\tilde{L}_{i-1}$ and $\tilde{L}_i$. $[\tilde{p}_{i-1, i}]$ corresponds to $\theta_{n_{i-1,i}}[\tilde{p}_{i-1,i}]$ in the B-model. The intersection of $L_{d}$ and $L_0$ is determined as $\tilde{p}_{0,d}=\frac{1}{n_{0,d}}\sum_{i=1}^d n_{i-1,i} \tilde{p}_{i-1,i}$. $[\tilde{p}_{0,d}]$ is a basis of the $m_d(p_{0,1},p_{1,2}, \cdots, p_{d-1,d})$, which corresponds to $\theta_{n_{0,d}}[\tilde{p}_{0,d}]$ in the B-model. The area of the convex surrounded by $L_0, \cdots, L_d$ is determined as $S_d(\tilde{p}_{0,1}, \cdots, \tilde{p}_{d-1, d})
=\frac{1}{2}(\sum_{i=1}^d n_{i-1,i} \tilde{p}_{i-1, i}(\tilde{p}_{i-1, i}-1) 
-n_{0,d}\tilde{p}_{0, d}(\tilde{p}_{0, d}-1))$. $\exp(2 \pi i \tau S_d(\tilde{p}_{0,1}, \cdots, \tilde{p}_{d-1, d}))$ is the coefficient of the $m_d(p_{0,1},p_{1,2}, \cdots, p_{d-1,d})$, which is equal to the coefficient of the $m_d(p_{0,1},p_{1,2}, \cdots, p_{d-1,d})$ in the B-model as one can see in Theorem \ref{producttheorem}, \ref{correct}.

We are going to show that the Feynman diagrams, which determine the $A_{\infty}$ structure of the B-model, coincide with the tropical Morse trees, which determine the $A_{\infty}$ structure of the A-model. That is, both the $A_{\infty}$ structures coincide. The tropical Morse trees $\tilde{\phi}$ are continuous maps that satisfy the following 
conditions,\footnote{
The condition (6) in \cite{BraneMirror} determines not the tropical Morse trees but the length of the metric ribbon trees from the tropical Morse trees.}
 from the metric ribbon trees $S$ to the universal cover $B$ ($\ni \tilde{y}$) of $\bold{R}/\bold{Z}$ ($\ni y$) \cite{phD, Progress, BraneMirror}.
\begin{itemize}
\item[(0)] We attach $n_{i-1, i}$ to the i-th external incoming edge and attach the sum of all numbers labelling the edges coming into a vertex to the edge coming out of the vertex. That is, the numbers $n$ are preserved at vertices. We define an affine displacement vector $v$ on each edge. 
\item[(1)] The coordinates of the external incoming and outgoing vertices are $\tilde{y}=\tilde{p}_{i-1,i}$ and $\tilde{y}=\tilde{p}_{0,d}$, respectively.
\item[(2)] All the edges and vertices of the metric ribbon trees map to the edges and vertices in the universal cover, respectively.
\item[(3)] $v=0$ at the external vertices. 
\item[(4)] $\Delta v=n \Delta \tilde{y}$ on the edges. The directions of $v$ and the edges coincide.
\item[(5)] $v$ are preserved at the vertices. 
\end{itemize}
 The conditions $(0), \cdots, (5)$ coincide with the properties of the Feynman diagrams in the B-model: (1) is satisfied in the B-model by Theorem \ref{correct}. (2) is automatically satisfied in the B-model. (0), (3), and (5) follow from the definition of the canonical form of the $\tilde{\theta}$. (4) coincides with the condition that  $\Delta v=n \Delta \tilde{p}$ and directions of the propagators are the same as those of $v$. 
\end{proof}
 Because the Feynman diagrams in the B-model coincide with the tropical Morse trees in the A-model, $\mathcal{F}(p_{0,d}; p_{0,1},p_{1,2}, \cdots, p_{d-1,d})
=
S^{trop}(p_{0,d}; p_{0,1},p_{1,2}, \cdots, p_{d-1,d})$, which is the moduli space of the tropical Morse trees. It is shown in \cite{phD, Progress, BraneMirror} that $S^{trop}(p_{0,d}; p_{0,1},p_{1,2},\cdots, p_{d-1,d})$ is bijective\footnote{The bijection map is given as follows \cite{phD, Progress, BraneMirror}: If the tropical Morse trees are decomposed to the edges, the corresponding pseudo holomorphic curves are also decomposed, and thus, the moduli spaces are decomposed as $\mathcal{S}^{trop}= \bigcup_{e} \mathcal{S}_e^{trop}$ and $\mathcal{M}/PSL(2, \bold{R}) = \bigcup_{e} \mathcal{M}_e/PSL(2, \bold{R})$. If the decomposed tropical Morse trees are denoted as $\mathcal{S}_e^{trop} \ni \tilde{\phi}: S_e \to B$ defined by $s \mapsto \tilde{\phi}(s)$ $(0 \le s \le 1)$, the decomposed pseudo holomorphic curves are denoted as $\mathcal{M}_e/PSL(2, \bold{R}) \ni R_{\tilde{\phi}}: D_e \to TB$  defined by $(s, t) \mapsto (\tilde{\phi}(s), -n_e \tilde{\phi}(s) -t \frac{d}{ds}\tilde{\phi}(s) )$ $(0 \le t \le 1)$, where $\frac{d}{ds}\tilde{\phi}(s)=v(s)$. Then, the map between these moduli spaces is given by $F: \mathcal{S}_e^{trop} \to \mathcal{M}_e/PSL(2, \bold{R})$ defined by $\tilde{\phi} \mapsto R_{\tilde{\phi}}$, that is, $\tilde{\phi}(s) \mapsto (\tilde{\phi}(s), -n_e \tilde{\phi}(s) -t \frac{d}{ds}\tilde{\phi}(s) )$.} to the moduli space $\mathcal{M}(p_{0,d}; p_{0,1},p_{1,2}, \cdots, p_{d-1,d})/PSL(2, \bold{R})$ of the pseudo holomorphic maps from a disk with $d+1$ marked points to the symplectic torus $TB$ in the A-model in one dimension.
If a topology, for example Gromov topology \cite{Gromov}, is defined on $\mathcal{M}(p_{0,d}; p_{0,1},p_{1,2}, \cdots, p_{d-1,d})/PSL(2, \bold{R})$, a topology is uniquely induced on $S^{trop}(p_{0,d}; p_{0,1},p_{1,2},$ $ \cdots, p_{d-1,d})$ by the bijective map\footnote{For example, if $\mathcal{M}_e/PSL(2, \bold{R})$ is equipped with the topology of the functional space, equivalently, compact-open topology, the same topology is induced to $\mathcal{S}_e^{trop}$, and $\mathcal{S}_e^{trop} \cong \mathcal{M}_e/PSL(2, \bold{R})$. By patching them together, $\mathcal{S}^{trop} \cong \mathcal{M}/PSL(2, \bold{R})$.}. 
As a result, we obtain
\begin{theorem}[topological space in homological mirror symmetry]

$\mathcal{F}(p_{0,d}; p_{0,1},p_{1,2}, \cdots, p_{d-1,d})$ in the B-model is homeomorphic to $\mathcal{M}(p_{0,d}; p_{0,1},p_{1,2}, \cdots, p_{d-1,d})/PSL(2, \bold{R})$ in the A-model.\end{theorem}

\vspace{1cm}

\section{Conclusion and Discussion}
\setcounter{equation}{0}
In this paper, we have defined and explicitly constructed $m_d$ that form $A_{\infty}$-category of the B-model on an elliptic curve. The way to define $m_d$ is summarized as follows. The morphisms in the DG-category of this model are represented by theta functions with characteristics. From the products of the two theta functions, we can derive Feynman rules in the space of the characteristics, in other words, the configuration space of the morphisms. $m_d$ are defined by the products of $d$ theta functions with the characteristics in the case that the moduli space of the Feynman diagrams is restricted to zero-dimensional. This restriction is necessary for $m_d$ to satisfy the $A_{\infty}$ relations. This $m_d$ algebra coincides with the $m_d$ algebra in the Fukaya category of the A-model on the corresponding symplectic torus. Thus, the $A_{\infty}$-category formed by these $m_d$ is equivalent to the Fukaya category as an $A_{\infty}$-category. We have also shown that the moduli space of the Feynman diagrams in the B-model is homeomorphic to the moduli space of the pseudo holomorphic curves in the A-model. 

 The way to define $m_d$ is naturally generalized in the case of the B-model on $2n$-dimensional complex manifolds: We derive Feynman rules in the configuration space of the morphisms from the products of the two morphisms in the DG-category of the B-model. Then we define $m_d$ by the products of $d$ morphisms with the configurations in the case that the moduli space of the Feynman diagrams is restricted to zero-dimensional. We conjecture that these $m_d$ satisfy the $A_{\infty}$ relations and form an $A_{\infty}$-category, which is equivalent to the Fukaya category of the A-model on the corresponding $2n$-dimensional symplectic manifold, as an $A_{\infty}$-category. We also conjecture that the moduli space of the Feynman diagrams in the B-model on the $2n$-dimensional complex manifolds is homeomorphic to the moduli space of the pseudo holomorphic curves in the A-model on the corresponding $2n$-dimensional symplectic manifold. These moduli spaces determine the $A_{\infty}$ structure of the both models.

%\section*{Note Added}
%On the day when we submitted this manuscript to the arXiv, a related paper \cite{} appeared on the arXiv.

\end{document}